\documentclass[letterpaper, 10 pt, conference]{ieeeconf}
\IEEEoverridecommandlockouts

\usepackage{amsmath,dsfont,amssymb}
\usepackage{amsthm}
\usepackage{mathtools}
\usepackage{graphicx}
\usepackage{amsfonts}
\usepackage{float}

\DeclareMathOperator{\Id}{\mathbf{I}}

\DeclareMathOperator{\Sat}{\mathbf{Sat}}

\DeclareMathOperator{\Pfx}{\Phi_{f_\pi}}
\usepackage[hidelinks]{hyperref}
\usepackage{algorithmicx}
\usepackage{algorithm}
\usepackage{algpseudocode}
\usepackage{lipsum}
\usepackage{subfig}
\usepackage{amssymb}
\usepackage{bbm}
\usepackage{float}
\usepackage{balance}
\usepackage{color}
\theoremstyle{plain}
\newtheorem{thm}{\textbf{Theorem}}
\newtheorem{lem}{\textbf{Lemma}}
\newtheorem{prop}{\textbf{Proposition}}

\theoremstyle{definition}
\newtheorem{defn}{\textbf{Definition}}

\newtheorem{ass}{\textbf{Assumption}}
\theoremstyle{remark}
\newtheorem{rem}{\textbf{Remark}}

\allowdisplaybreaks
\usepackage{algpseudocode}
\algdef{SE}[DOWHILE]{Do}{DoWhile}{\algorithmicdo}[1]{\algorithmicwhile\ #1}%

\usepackage[normalem]{ulem}

\begin{document}

\title{\LARGE \bf
Backup Control Barrier Functions: Formulation and Comparative Study
}
\author{Yuxiao Chen$^{1}$, Mrdjan Jankovic$^{2}$, Mario Santillo$^{2}$,
and Aaron D. Ames$^{1}$
\thanks{$^{1}$ Yuxiao Chen and Aaron D. Ames are with Department of Mechanical and Civil Engineering, California Institute of Technology, Pasadena, CA, USA {\tt\small chenyx,ames@caltech.edu}}
\thanks{$^{2}$ Mario Santillo and Mrdjan Jankovic are with Ford Research and Advanced Engineering, Dearborn, MI, USA {\tt\small msantil3,mjankov1@ford.com}}
}

\maketitle
\thispagestyle{empty}

\begin{abstract}
The backup control barrier function (CBF) was recently proposed as a tractable formulation that guarantees the feasibility of the CBF quadratic programming (QP) via an implicitly defined control invariant set. The control invariant set is based on a fixed backup policy and evaluated online by forward integrating the dynamics under the backup policy. This paper is intended as a tutorial of the backup CBF approach and a comparative study to some benchmarks. First, the backup CBF approach is presented step by step with the underlying math explained in detail. Second, we prove that the backup CBF always has a relative degree 1 under mild assumptions. Third, the backup CBF approach is compared with benchmarks such as Hamilton Jacobi PDE and Sum-of-Squares on the computation of control invariant sets, which shows that one can obtain a control invariant set close to the maximum control invariant set under a good backup policy for many practical problems.
\end{abstract}

\section{Introduction}\label{sec:intro}

Control barrier functions (CBF) \cite{ames2014control,wieland2007constructive} were proposed as a method that enforces constraints on dynamic systems, which typically works as a supervisory controller on top of a legacy controller. To guarantee the satisfaction of the safety constraints, a CBF quadratic program (QP) is solved online. While CBF is getting increasingly popular due to its simple implementation and strong guarantee, the construction of CBF is sometimes overlooked.

To guarantee that the CBF QP is always feasible, a control invariant set is needed, which is defined as a set in which any trajectory of the dynamic system can stay indefinitely. The concept of control invariant sets has been studied under various background and names, such as viability kernel \cite{aubin2011viability}, infinite time reachable set \cite{bertsekas1972infinite}, and various methods have been proposed to compute the control invariant set depending on the system dynamics, see \cite{blanchini1999set} for an overview. Unfortunately, the computation of control invariant sets is notoriously difficult. Even for simple cases such as linear or polynomial dynamics, computation tools such as Minkowski operations \cite{rakovic2004computation}, robust linear program \cite{chen2018data}, and Sum-of-Squares \cite{korda2014convex} do not scale well. For general nonlinear dynamic systems, the standard tool for computing invariant sets is Hamilton Jacobi PDE \cite{mitchell2005time}, which typically cannot scale beyond systems with state dimensions 4 due to the exponential complexity.

Due to the difficulty of synthesizing proper control barrier functions based on control invariant sets, CBF QP has been implemented without control invariant sets. One simple treatment is to assume infinite actuation power \cite{taylor2020learning}, in which case the feasibility of the CBF QP can be guaranteed when a simple sign condition is satisfied. To be specific, if one can show that when the Lie derivative of the CBF w.r.t. the input dynamics is zero, the CBF condition is satisfied by the intrinsic dynamics, the CBF QP is always feasible under infinite actuation. Another simplification is to assume that the velocity instead of the acceleration is under control \cite{lindemann2018control}, which can guarantee the feasibility of the CBF QP for constraints on the position since the CBF QP can always pick a velocity pointing away from the constraint. In essence, assuming direct velocity control is similar to assuming infinite actuation as an instantaneous change of velocity requires infinite force. Obviously, these assumptions are not true in practice, and the safety of the system is subject to parameter tuning. The simplified approaches might work for low-speed cases as the change of velocity is not severe, but will not work in general for highly dynamic applications.

Another issue of CBF without a control invariant set is the relative degree. For example, a typical vehicle/robot model with acceleration input is a second (or higher) order model with the position states and velocity states. If one directly uses the constraint on the position as a control barrier function, the relative degree of the CBF is 2, and its Lie derivative does not contain the acceleration input. Several solutions have been proposed for the high relative degree, such as Input-Output Linearization \cite{nguyen2016exponential,xu2018constrained} and backstepping type formulations \cite{xiao2019control}. Again, the CBF performance is subject to parameter tuning.

To the best of our knowledge, under limited input, there is not a generic method that guarantees the feasibility of the CBF QP without a control invariant set. All of the above-mentioned methods rely on heuristics and tuning to work in practice. Even when the CBF works well in the test cases, there is no guarantee of performance for cases not included in the test.

The idea of control barrier functions based on backup controllers (referred to as backup CBF for the remainder of this paper) \cite{gurriet2020scalable} was proposed based on a simple observation that extends a small control invariant set, which is typically easier to obtain, to a larger control invariant set by fixing a backup controller. The backup CBF guarantees the feasibility of the CBF QP and circumvents the difficult computation of a control invariant set by implicitly representing the control invariant set. However, the implicit representation calls for online integration of the dynamics, and the resulting control invariant set is typically suboptimal in the sense that it is not the maximum control invariant set. This paper is intended as a tutorial to the backup CBF approach by demonstrating the method step by step on some simple examples. We also provide explanations and visualizations on the advantage and disadvantages of the backup CBF over some benchmark formulations.

\section{Preliminaries and a motivating example}\label{sec:prelim}
\subsection{Control barrier functions}
We begin by a brief review of control barrier functions. Consider the following control affine system:
\begin{equation}\label{eq:dyn}
  \dot{x}=f(x)+g(x)u,\quad x\in\mathbb{R}^n,u\in\mathcal{U}\subseteq \mathbb{R}^m,
\end{equation}
where $x$ is the system state and $u$ is the input. Suppose there exists a function $h:\mathbb{R}^n\to\mathbb{R}$ that satisfies
\begin{equation}\label{eq:CBF}
  \begin{aligned}
&\forall~ x \in {\mathcal{X}_0},&h(x) \ge 0\\
&\forall~ x \in {\mathcal{X}_d},&h(x) < 0\\
&\forall~ x \in \left\{ {x\mid h(x) \ge 0} \right\}, &\exists~ u \in \mathcal{U}\;
\mathrm{s.t.}~\dot h + \alpha \left( h \right) \ge 0,
\end{aligned}
\end{equation}
where $\mathcal{X}_0$ is the set of initial states and $\mathcal{X}_d$ is the danger set that we want to keep the state away from. $\alpha(\cdot)$ is a class-$\mathcal{K}$ function, i.e., $\alpha(\cdot)$ is strictly increasing and satisfies $\alpha(0)=0$. Then $h$ is called control barrier function, and for any legacy controller, the CBF controller is a supervisory controller that enforces the state to stay inside $\left\{ {x\mid h(x) \ge 0} \right\}$ with the following quadratic programming:
\begin{equation}\label{eq:CBF_QP}
  \begin{aligned}
u^\star = &\mathop {\arg \min }\limits_{u \in \mathcal{U}} \left\| {u - {u^0}} \right\|^2\\
\mathrm{s.t.}~&\nabla h\cdot f\left( {x,u} \right) + \alpha \left( h \right) \ge 0,
\end{aligned}
\end{equation}
where $u^0$ is the input of the legacy controller.

The third condition in \eqref{eq:CBF} ensures that \eqref{eq:CBF_QP} is always feasible, yet it is difficult to find an $h$ that satisfies it. For clarity, we refer to a function that satisfies the first two conditions in \eqref{eq:CBF} a CBF candidate, a function that satisfies all three conditions a valid CBF.

The CBF condition is closely tied to the concept of a control invariant set, which is defined as follows.
\begin{defn}\label{def:inv}
  A set $\mathcal{S}$ is a control invariant set if there exists a control law $\pi:\mathbb{R}^n\to\mathcal{U}$ such that for all initial condition $x(0)\in\mathcal{S}$, $\forall t\ge0, x(t)\in\mathcal{S}$.
\end{defn}
It is straightforward to see that for a valid CBF, $\{x|h(x)\ge0\}$ is a control invariant set. On the other direction, suppose for a CBF candidate $h$, its 0-level set is a control invariant set. Immediately from Definition \ref{def:inv},
\begin{equation*}
  h(x)=0\to \exists u\in\mathcal{U}~\mathrm{s.t.} \dot h(x,u)\ge 0,
\end{equation*}
otherwise the state will exit $\mathcal{S}$, which contradicts Definition \ref{def:inv}. Then $h$ can be shown to be a valid control barrier function by picking an $\alpha$ large enough given some continuity condition (Lipschitz continuity of $h$ and $\dot{h}$). As a result, the third condition in \eqref{eq:CBF} is also referred to as the set invariance condition.

\subsection{Double integrator example}
As a motivating example, consider a simple double integrator with
\begin{equation*}
  \dot{x}=\begin{bmatrix}
    \dot{s} \\
    \dot{v}
  \end{bmatrix}=\begin{bmatrix}
                  v \\
                  u
                \end{bmatrix},\quad u\in[-u_{\max},u_{\max}],v\in[-v_{\max},v_{\max}].
\end{equation*}
The safety constraint is $X\le C$ with a constant $C$. Suppose one directly takes the safety constraint as a CBF candidate: $h_0(x)=C-s$. There are two issues. First, $h_0$ has relative degree 2, i.e., $\dot{h}_0$ is not a function of $u$. Second, the set $\{x|h_0(x)\ge 0\}$ is not control invariant. Consider the case where $s=C,v>0$, the limited input cannot stop the state from crossing into the danger set due to inertia.

For this simple system, it is widely known that a simple valid CBF exists:
\begin{equation}\label{eq:DI_h}
  h(x)=C-s-\mathds{1}_{v> 0}\frac{v^2}{2u_{\max}},
\end{equation}
which is the safety constraint combined with the minimum stopping distance. When $h(x)\ge 0$, one can always apply the maximum deceleration $u=-u_{\max}$ until $v=0$ and the safety constraint remains satisfied. Fig. \ref{fig:DI} shows the difference between $h$ and $h_0$ where $h_0$ includes a part of the state space from which the safety constraint will eventually be violated.

\begin{figure}
  \centering
  \includegraphics[width=0.7\columnwidth]{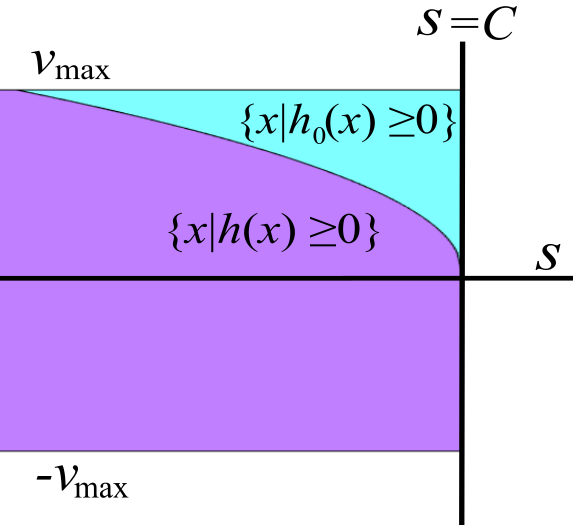}
  \caption{Difference between $h$ and $h_0$ for the double integrator example}\label{fig:DI}
\end{figure}

In fact, the above CBF is exactly the backup CBF with the backup policy being
\begin{equation*}
  \pi(x)=-\mathds{1}_{v> 0}u_{\max}.
\end{equation*}
We shall show later that the backup policy $\pi$ takes the system to a small control invariant set $\mathcal{S}_0=\{x|v=0\}$, and $\{x|h(x)\ge 0\}$ with $h$ in \eqref{eq:DI_h} is the set of initial conditions that can be brought to $\mathcal{S}_0$ while satisfying the safety constraint. In this view, $h$ is induced from the backup policy $\pi$.

As mentioned in the Introduction, computing a control invariant set is very difficult, especially for high-dimensional nonlinear systems. We shall show that for complicated dynamic systems, a valid CBF can be fairly easily obtained from a backup policy. Though the CBF induced from a backup policy might not have the nice closed form as in the double integrator case, the CBF QP can be solved online and is guaranteed to be feasible.
\subsection{Notation and preliminaries}
Before going into the detail, some important notations are reviewed. Given the dynamic system in \eqref{eq:dyn}, a control policy $\pi:\mathbb{R}^n\to\mathcal{U}$, the closed-loop dynamics under $\pi$ is $\dot{x}=f_\pi(x)=f(x)+g(x)\pi(x)$. We let $\Pfx:\mathbb{R}^n\times(-\infty,\infty)\to\mathbb{R}^n$ denote the flow map, i.e., $\Pfx(x_0,t)$ denotes the solution $x(t_0+t)$ to the Initial Value Problem (IVP) at time $t_0+t$ with $x(t_0)=x_0$ under the dynamics $f_\pi$. Note that since \eqref{eq:dyn} is time invariant, the initial time $t_0$ is irrelevant. Moreover, the flow map is additive in $t$:
\begin{equation}\label{eq:additive}
  \Pfx(\Pfx(x,t_1),t_2)=\Pfx(x,t_1+t_2).
\end{equation}


Since the flow map can be written as
\begin{equation}\label{eq:flow}
  \Pfx(x_0,t)=x_0+\int_{0}^{t} f_\pi(x(\tau))d\tau,
\end{equation}
it is differentiable w.r.t. both $t$ and $x$. By inspection, $\frac{\partial \Pfx(x_0,t)}{\partial t}=f_\pi(\Pfx(x_0,t))$. To obtain the partial derivative over $x_0$, take derivative on both sides:
\begin{equation*}
  \frac{\partial \Pfx(x_0,t)}{\partial x_0}=\Id+\int_{0}^{t} \frac{d f_\pi}{d x} \frac{\partial x(\tau)}{\partial x_0} d\tau,
\end{equation*}
from \cite{seywald2003desensitized}, we can define $Q(t)=\frac{\partial \Pfx(x_0,t)}{\partial x_0}$ as the sensitivity Jacobian, and taking derivative over $t$ on both sides yields
\begin{equation}\label{eq:ja}
  \dot{Q}(t)=\frac{d f_\pi(x(t))}{d x(t)} Q(t),
\end{equation}
where $x(t)=\Pfx(x_0,t)$. This fact will play a key role on the derivation of the CBF condition.

\begin{prop}\label{prop:deriv_inv}
For all $\tau\ge t$,
\begin{equation}\label{eq:deriv_inv}
\frac{d\Phi(x(t),\tau-t)}{dt}|_{u(t)=\pi(t)}=0.
\end{equation}
\end{prop}
\begin{proof}
Note that the flow map can be written as
\begin{equation*}
  \Pfx(x(t),\tau-t)=x(t)+\int_{t}^{\tau} f_\pi(x(\overline{\tau}))d\overline{\tau},
\end{equation*}
taking the derivative over $t$ gives
\begin{equation*}
  \frac{d\Phi(x(t),\tau-t)}{dt} = \dot{x}(t)-f_\pi(x(t)),
\end{equation*}
and the result follows as $\dot{x}|_{u=\pi(x)}=f_\pi(x)$.

\end{proof}
The physical interpretation of Proposition \ref{prop:deriv_inv} is that the flow under the backup strategy $\pi$ would not change with time if the current state follows $\pi$.
\section{Backup CBF}\label{sec:backup}
This section presents the backup CBF step by step.
\subsection{Definition and CBF QP}
Consider a safe region $\mathcal{C}\doteq\{x|h^\mathcal{C}\ge 0\}\subseteq \mathbb{R}^n$, a control invariant set $\mathcal{S}_0\doteq\{x|h^\mathcal{S}\ge 0\}\subseteq \mathcal{C}$ with $h^\mathcal{C},h^\mathcal{S}:\mathbb{R}^n\to\mathbb{R}$ differentiable. Given a backup controller $\pi$, define the $T$-time constrained reachable set as
\begin{equation}\label{eq:reach}
\begin{aligned}
  \mathcal{S}&=\mathcal{R}(\mathcal{S}_0,\mathcal{C},f_\pi,T)\\
  &\doteq\{x\in\mathbb{R}^n|\Pfx(x,T)\in\mathcal{S}_0\wedge \forall t\in[0,T],\Pfx(x,t)\in\mathcal{C}\},
\end{aligned}
\end{equation}
which contains the states from which the flow under $f_\pi$ will stay in $\mathcal{C}$ and reach $\mathcal{S}_0$ by $T$. The theoretical root of the backup CBF approach is the following simple observation.
\begin{thm}\cite{gurriet2018online}\label{thm:inv}
   Given an initial control invariant set $\mathcal{S}_0$, a safe region $\mathcal{C}\subseteq \mathbb{R}^n$, and a backup controller $\pi$, for all $T>0$, the $T$-time constrained reachable set $\mathcal{S}$ defined in \eqref{eq:reach} is a control invariant set and $\mathcal{S}_0\subseteq\mathcal{S}\subseteq\mathcal{C}$.
\end{thm}
The proof can be found in \cite{gurriet2018online}, we provide a brief proof for completeness.
\begin{proof}
  By definition, $\mathcal{S}_0\subseteq\mathcal{S}$, and $\mathcal{S}\in\mathcal{C}$ since $\Pfx(x,0)\in\mathcal{C}$. By Definition \ref{def:inv}, there exists $\pi_0$ that keeps $\mathcal{S}_0$ invariant. Define $\pi'(x)=\left\{ {\begin{array}{*{20}{c}}
  {{\pi _0}(x),x \in {\mathcal{S}_0}} \\
  {\pi (x),x \notin {\mathcal{S}_0}}
\end{array}} \right.$, it can be verified that $\pi'$ keeps $\mathcal{S}$ invariant.
\end{proof}
\begin{figure}
  \centering
  \includegraphics[width=1\columnwidth]{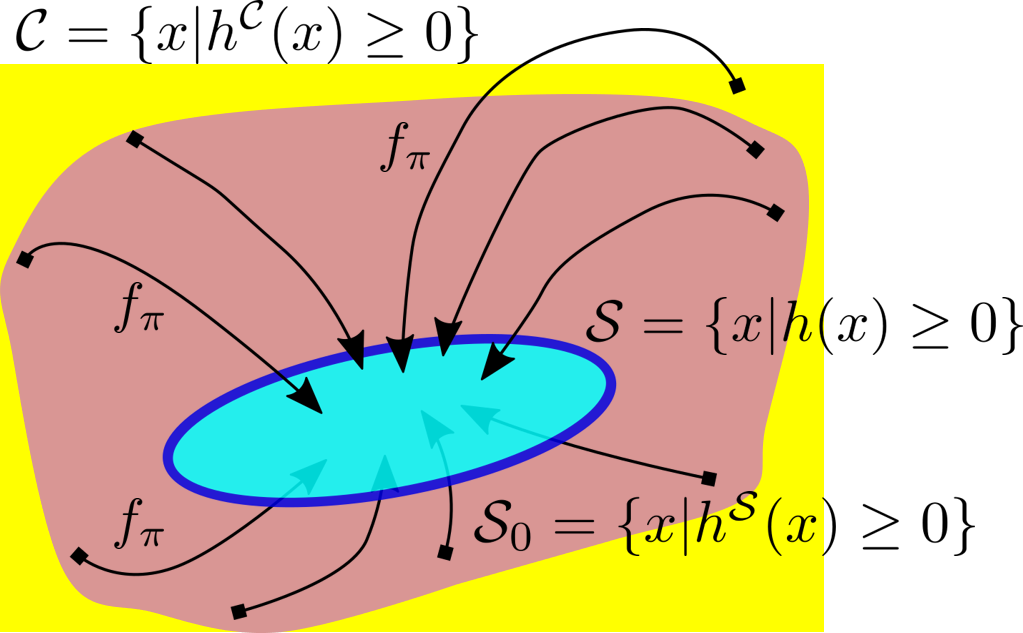}
  \caption{Enlarging a known control invariant set $\mathcal{S}_0$ with the backup policy $\pi$}\label{fig:backup}
\end{figure}

Fig. \ref{fig:backup} shows theorem \ref{thm:inv} pictorially. The yellow square is the state constraint $\mathcal{C}$, given a small known control invariant set $\mathcal{S}_0$ (the blue ellipse), all states that can be driven to $\mathcal{S}_0$ under $f_\pi$ (the closed-loop dynamics under the backup policy $\pi$) while staying inside $\mathcal{C}$ in the meantime forms a larger control invariant set $\mathcal{S}$ shown in brown. Note that the point on the top right is not inside $\mathcal{S}$ because the trajectory under $f_\pi$ is not completely contained $\mathcal{C}$.
 
\begin{ass}\label{ass:ci}
For simplicity, we assume that the backup controller $\pi$ keeps $\mathcal{S}_0$ invariant, and the class-$\mathcal{K}$ function $\alpha$ is selected so that $\forall x\in \mathcal{S}_0, \nabla h^\mathcal{S}f_\pi(x)+\alpha(h(x))\ge 0$.
\end{ass}
 Regarding the horizon $T$, the following lemma is true.
\begin{lem}\label{lem:mono}
  For all $T_1,T_2>0$, let $\mathcal{S}=\mathcal{R}(\mathcal{S}_0,\mathcal{C},f_\pi,T_1)$, then $\mathcal{R}(\mathcal{S},\mathcal{C},f_\pi,T_2)=\mathcal{R}(\mathcal{S}_0,\mathcal{C},f_\pi,T_1+T_2)$.
\end{lem}
\begin{proof}
  By definition, $\forall x\in \mathcal{R}(\mathcal{S},\mathcal{C},f_\pi,T_2),$ $\Pfx(x,T_2)\in\mathcal{S}=\mathcal{R}(\mathcal{S}_0,\mathcal{C},f_\pi,T_2)$. By \eqref{eq:additive}, $\Pfx(x,T_1+T_2)=\Pfx(\Pfx(x,T_2),T_1)\in\mathcal{S}_0$. And by \eqref{eq:reach}, the trajectory within $[0,T_1+T_2]$ is contained in $\mathcal{C}$, and the conclusion follows.
\end{proof}
Lemma \ref{lem:mono} states that $\mathcal{R}(\mathcal{S}_0,\mathcal{C},f_\pi,T)$ monotonically increases with $T$ in the set inclusion sense.

With a control invariant set $\mathcal{S}$ defined, a control barrier function $h$ can be defined.
\begin{lem}\label{lem:CBF}
  $\mathcal{S}$ is the 0-level set of the following function
  \begin{equation}\label{eq:cbf}
    h(x)=\min\{\mathop{\min}\limits_{t'\in[0,T]}h^\mathcal{C}(\Phi_{f_\pi}(x,t')),h^\mathcal{S}(\Phi_{f_\pi}(x,T))\}.
  \end{equation}
\end{lem}

\begin{proof}
First notice that by the continuity of the flow function $\Phi_{f_\pi}$ and the $\min$ function, $h$ is continuous. For all $x\in\mathcal{S}$, by definition, under the backup strategy $\pi$, the state evolution $\Phi_{f_\pi}(x,t')$ would satisfy the constraint and reach $\mathcal{S}_0$ at time $T$, therefore $h(x)\ge 0$. On the other hand, for all $x\notin\mathcal{S}$, under the backup strategy $\pi$, the state evolution either violates the state constraint at some $t'$, i.e., $\exists t'\in[0,T], h^\mathcal{C}(\Phi_{f_\pi}(x,t'))<0$, or does not reach $\mathcal{S}_0$ within the horizon $T$, i.e., $h^\mathcal{S}(\Phi_{f_\pi}(x,T))<0$, indicating that $h(x)<0$. Therefore, $\mathcal{S}=\{x|h(x)\ge 0\}$.
\end{proof}
To implement the CBF $h$, $\dot{h}$ is needed. First $h$ needs to be written as a function of time:
\begin{equation}\label{eq:ht}
  h(t)=\min\{\mathop{\min}\limits_{\tau\in[t,t+T]}h^\mathcal{C}(\Phi_{f_\pi}(x(t),\tau-t)),h^\mathcal{S}(\Phi_{f_\pi}(x(t),T))\}.
\end{equation}
 Since $h$ is defined as the minimum of multiple functions, it may not be differentiable. Even in the differentiable case, the CBF condition may turn out to be nonconvex in the control input $u$. Instead, the CBF condition is enforced on every $\tau$, which is a sufficient condition for original CBF condition.

First consider $h^\mathcal{S}(\Pfx(x(t),T))$. By the chain rule,
\begin{equation*}
\resizebox{1\columnwidth}{!}{$
  \frac{d h^\mathcal{S}(\Pfx(x(t),T))}{dt}=\frac{d h^\mathcal{S}}{dx} \frac{\partial \Pfx(x(t),T)}{\partial x} (f(x(t))+g(x(t))u(t)),
  $}
\end{equation*}
where the sensitivity Jacobian $\frac{\partial \Pfx(x(t),T)}{\partial x}$ can be calculated with \eqref{eq:ja}. The computation of $h^\mathcal{C}(\Phi_{f_\pi}(x(t)),\tau-t)$ is slightly different:
\begin{equation*}
\begin{aligned}
&\frac{d h^\mathcal{C}(\Pfx(x(t),\tau-t))}{dt}\\
=&\frac{d h^\mathcal{S}}{dx}(\frac{\partial\Pfx(x(t),\tau-t)}{\partial x}\dot{x}-\frac{\partial\Pfx(x(t),\tau-t)}{\partial t} )\\
=&\frac{d h^\mathcal{S}}{dx}(\frac{\partial\Pfx(x(t),\tau-t)}{\partial x}\dot{x}-f_\pi(\Pfx(x(t),\tau-t)) ),
\end{aligned}
\end{equation*}
where $\dot{x}=(f(x(t))+g(x(t))u(t))$.
By Proposition \ref{prop:deriv_inv}, $\frac{d\Pfx(x(t),\tau-t)}{dt}|_{u=\pi(x)}=0$, indicating that $\frac{d h^\mathcal{C}(\Pfx(x(t),\tau-t))}{dt}=0$ if $u(t)=\pi(x(t))$.
\begin{rem}
The difference between the derivative of $h^\mathcal{S}(\Pfx(x(t),T))$ and $h^\mathcal{C}(\Pfx(x(t),\tau-t))$ is that the former contains two parts, the change of the future state following the backup strategy due to the change of current state, and the derivative due to $t+T$ increasing with $t$; whereas the latter only contains the first part.
\end{rem}
The CBF condition is $\dot{h}+\alpha(h)\ge 0$.  We shall impose this condition on $h^\mathcal{S}(\Pfx_{f_\pi}(x,T))$ and $h^\mathcal{C}(\Pfx_{f_\pi}(x,\tau-t))$ for every $\tau\in[t,t+T]$ instead of only on the $\tau$ minimizing $h^\mathcal{C}(\Pfx_{f_\pi}(x,\tau-t))$, and the CBF QP is then
\begin{equation}\label{eq:CBF_QP_backup}
  \begin{aligned}
  \mathop{\min}\limits_{u\in\mathcal{U}} &||u-u^0||\\
  \mathrm{s.t.}& \forall \tau\in[t,t+T]\\
  &\frac{d h^\mathcal{C}(\Pfx(x,\tau-t))}{dt}(x,u)+\alpha(h^\mathcal{C}(\Pfx(x(t),\tau-t)))\ge 0\\
  &\frac{d h^\mathcal{S}(\Pfx(x(t),T))}{dt}(x,u)+\alpha(h^\mathcal{S}(\Pfx(x(t),T)))\ge 0.
  \end{aligned}
\end{equation}

\begin{prop}
The constraint in \eqref{eq:CBF_QP_backup} is a sufficient condition for $\dot{h}+\alpha(h)\ge0$.
\end{prop}
\begin{proof}
For notational simplicity, let $\bar{\xi}^\mathcal{C}(x,\tau)\doteq h^\mathcal{C}(\Pfx(x,\tau-t))$, $\xi^\mathcal{C}(x)\doteq \mathop{\min}\limits_{\tau\in[t,t+T]} \bar{\xi}^\mathcal{C}(x,\tau) $, and $\xi^\mathcal{S}\doteq h^\mathcal{S}(\Pfx(x,T))$, then $h(x) = \min\{\xi^\mathcal{C}(x),\xi^\mathcal{S}(x)\}$. First notice that
\begin{equation*}
\resizebox{1\columnwidth}{!}{$
\frac{d \xi^\mathcal{C}(x(t))}{dt}+\alpha(\xi^\mathcal{C}(x(t)))
\ge \mathop{\min}\limits_{\tau\in[t,t+T]}\left[\dot{\bar{\xi}}^\mathcal{C}(x(t),\tau)+\alpha(\bar{\xi}^\mathcal{C}(x(t),\tau))\right]\ge 0,$}
\end{equation*}
then
  \begin{equation*}
  \resizebox{1\columnwidth}{!}{$
  \begin{aligned}
  \dot{h}+\alpha(h)=&\frac{d \min\{\xi^\mathcal{C}(x(t)),\xi^\mathcal{S}(x(t))\}}{dt}+\alpha(\min\{\xi^\mathcal{C}(x(t)),\xi^\mathcal{S}(x(t))\})\\
  \ge &\min\{\dot{\xi}^\mathcal{C}(x(t))+\alpha(\xi^\mathcal{C}(x(t))),\dot{\xi}^\mathcal{S}(x(t))+\alpha(\xi^\mathcal{S}(x(t)))\}\ge 0.
  \end{aligned}
  $}
  \end{equation*}
\end{proof}

In practice, $\mathop{\min}\limits_{\tau\in[t,t+T]}h^\mathcal{C}(\Phi_{f_\pi}(x(t)),\tau-t)$ may be difficult to evaluate, so a finite set of $t=\tau_0<\tau_1...<\tau_N=t+T$ is used instead of $[t,t+T]$, and the approximation error can be bounded given the Lipschitz constants of the flow map.

\begin{thm}\label{thm:feas}
  Under Assumption \ref{ass:ci}, the backup CBF QP in \eqref{eq:CBF_QP_backup} is always feasible for $h\ge 0$.
\end{thm}
\begin{proof}
Since $h(x)\ge0$, $\forall \tau\in[t,t+T],h^\mathcal{C}(\Pfx(x(t),\tau-t))\ge 0$, $h^\mathcal{S}(\Pfx(x(t),T))\ge 0$. By proposition \ref{prop:deriv_inv}, if $u(t)=\pi(x(t))$, for all $\tau\in[t,t+T]$, $\frac{d h^\mathcal{C}(\Pfx(x(t),\tau-t))}{dt}=0$, $\frac{d h^\mathcal{S}(\Pfx(x(t),T))}{dt}=\frac{d h^\mathcal{S}}{dx}f_\pi(\Pfx(x(t),T))$. Immediately, the first constraint in \eqref{eq:CBF_QP_backup} is satisfied. By Assumption \ref{ass:ci}, since $\Pfx(x(t),T)\in\mathcal{S}_0$, $\nabla h^\mathcal{S}f_\pi(\Pfx(x(t),T))+\alpha(h^\mathcal{S}(\Pfx(x(t),T)))\ge0$, indicating that the second constraint in \eqref{eq:CBF_QP_backup} is satisfied. Since $\pi$ is a feasible control policy, $\pi(x(t))\in\mathcal{U}$, thus a feasible solution to \eqref{eq:CBF_QP_backup}.
\end{proof}

\subsection{Relative degree of backup CBFs}

Another benefit of the backup CBF is that it is always relative degree one given the system is weakly locally controllable. To show this, we first review some definitions.
\begin{defn}
  A function $h:\mathbb{R}^n\to\mathbb{R}$ is said to have relative degree $r$ with respect to the dynamic system in \eqref{eq:dyn} at a point $x_0$ if
  \begin{equation*}
  \resizebox{1\columnwidth}{!}{$
  \begin{aligned}
  &\mathcal{L}_g\mathcal{L}_f^k h(x)=0, &\forall x \text{ in a neighborhood of $x_0$, $\forall k\le r-2$,}\\
  & \mathcal{L}_g\mathcal{L}_f^{r-1} h(x_0)\neq 0 &
  \end{aligned}
  $}
  \end{equation*}
\end{defn}
\begin{figure}[t]
  \centering
  \includegraphics[width=0.85\columnwidth]{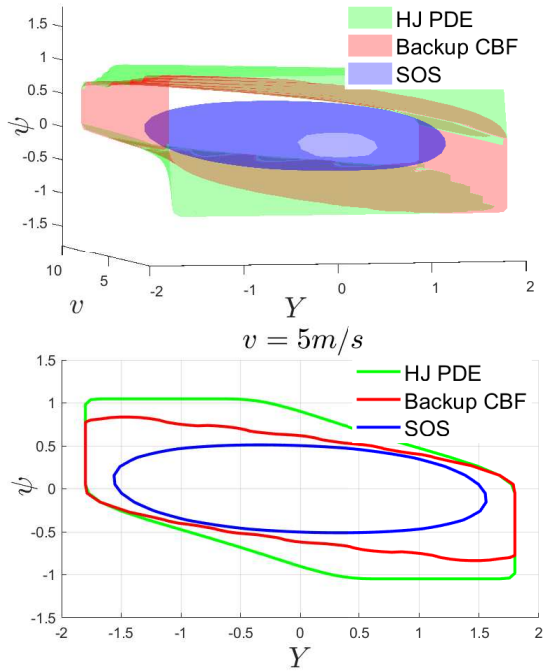}
  \caption{$\{x|h(x)\ge 0\}$ under CBF computed with Hamilton Jacobi PDE (green), backup CBF (red), and SOS (blue)}\label{fig:SOS_comp}
\end{figure}

\begin{defn}\cite{hermann1977nonlinear}
  Given a dynamic system, $x_1$ is accessible from $x_0$ via $\Omega\subseteq\mathbb{R}^n$ if there exists a control input signal $\mathbf{u}:[t_0,t_1]\to \mathcal{U}$ such that the trajectory starting at $x_0$ reaches $x_1$ at $t_1$ and $\forall t\in[t_0,t_1],x(t)\in \Omega$. The accessible set of a point $x_0$ under $\Omega$ is the set of all points accessible from $x_0$ under $\Omega$, denoted as $A_\Omega(x_0)$.
\end{defn}
\begin{defn}\cite{hermann1977nonlinear}
The dynamic system in \eqref{eq:dyn} is locally weakly controllable if for every $x_0\in\mathbb{R}^n$ and every neighborhood $\Omega$ of $x$, $A_\Omega(x_0)$ has a non-empty interior.
\end{defn}
\begin{thm}\label{thm:rel_deg}
  If the dynamic system \eqref{eq:dyn} is locally weakly controllable, then for an $h:\mathbb{R}^n\to\mathbb{R}$ that is not constant for any subset $\Omega\subseteq \mathbb{R}^n$ with a non-empty interior, $h\circ\Pfx(x,t)$ has relative degree 1 apart from singular points for all $t>0$.
\end{thm}
\begin{proof}
  We prove this by contradiction. Suppose there exists $x_0\in\mathbb{R}^n$ and a neighborhood $\Omega$ where $\forall x\in \Omega, \mathcal{L}_g (h\circ\Pfx)(x)=0$. Then we can find a small $t_0>0$ such that the flow from $x_0$ stays inside $\Omega$ within $t\in[0,t_0]$. Since the system is locally weakly controllable, $A_\Omega(x_0)$ is a compact set with a non-empty interior. With a slight abuse of notation, let $\Pfx(S,t)=\{x|x=\Pfx(x_0,t),x_0\in S\}$ be the image of the flow map acting on a subset $S\subseteq\mathbb{R}^n$. For any $t>0$, since we only consider regular points, $\Pfx(\cdot,t)$ has rank $n$, thus $\Pfx(A_\Omega(x_0),t)$ has non-empty interior. However, since $\mathcal{L}_g h(x)=0$ for all $x\in \Omega$, $h(x)$ is constant inside $\Pfx(A_\Omega(x_0),t)$, which contradicts the assumption that $h$ is not constant in any set with non-empty interior. Thus, $h\circ\Pfx(x,t+t_0)$ has relative degree 1. Since $t_0$ can be chosen arbitrarily small, $h\circ\Pfx(x,t)$ has relative degree 1 for all $t>0$.
\end{proof}

We cannot show that the CBF defined in \eqref{eq:cbf} has relative degree 1 due to the $\min$ function over $t\in[0,T]$, however, note that the constraints in the CBF QP in \eqref{eq:CBF_QP_backup} are on all $h^\mathcal{C}\circ\Pfx(x,\tau-t)$ with $\tau\in[t,t+T]$ and $h^\mathcal{S}\circ\Pfx(x,T)$. Therefore, as long as $h^\mathcal{S}$ and $h^\mathcal{C}$ are not constant for any subset with a non-empty interior, the CBF QP can be solved as if $h$ has relative degree 1.

Theorem \ref{thm:rel_deg} shows that the backup CBF is sufficient without any high order extension or backstepping, it naturally bridges the potentially high relative degree functions $h^\mathcal{C}$ and $h^\mathcal{S}$ with the input dynamics via the flow map.

\begin{figure}[t]
  \centering
  \includegraphics[width=0.9\columnwidth]{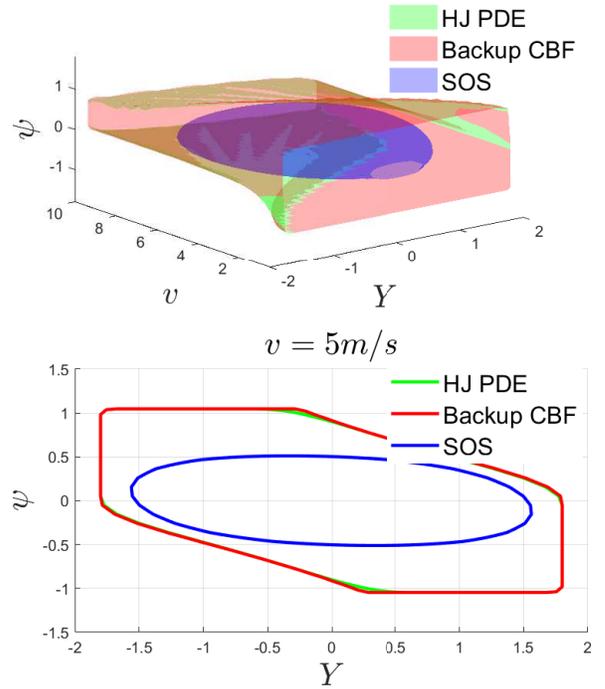}
  \caption{$\{x|h(x)\ge 0\}$ with an aggressive backup policy}\label{fig:aggressive}
\end{figure}

\section{Comparative study}\label{sec:comparison}
This section presents the comparative study of the backup CBF against some benchmark methods.

One major concern of the backup CBF is that since the control invariant set is induced by a fixed backup policy, how conservative is the resulting CBF? The conservatism of the CBF can be measured by the size of the set $\{x|h(x)\ge0\}$, which is the induced control invariant set. We compare the backup CBF with the Hamilton Jacobi (HJ) PDE result, which is a close approximation of the maximum control invariant set, and a Sum-of-Squares (SOS) result. The Hamilton Jacobi formulation follows \cite{mitchell2005time} and is computed with the level-set toolbox. In particular, the following HJ PDE is solved:
\begin{equation*}
 -\frac{\partial h}{\partial t}= H(x,p)= \max\limits_{u\in\mathcal{U}} p^\intercal f(x,u),
\end{equation*}
where $p=\frac{\partial h}{\partial x}$ is the momentum vector. $\{x|h(x,t)\ge 0\}$ is then the backward reachable set at time $t$. For a sufficiently large $t$, $\lim_{t\to\infty} h(x,t)$ is an inner approximation of the maximum control invariant set.

The SOS approach follows \cite{xu2017correctness,chen2019enhancing}, which computes a valid CBF under a fixed controller. For fairness, we chose the same backup controller $\pi$ for the backup CBF and the SOS program.

\subsection{Lane keeping with Dubin's car model}
To demonstrate the result, a simple lane keeping problem under Dubin's car model is considered:
\begin{equation}\label{eq:lane_keeping}
  \dot{x}=\begin{bmatrix}
    \dot{Y} &
    \dot{v} &
    \dot{\psi}
  \end{bmatrix}^\intercal=f(x,u)=\begin{bmatrix}
                  v \sin(\psi) &
                  a &
                  r
                \end{bmatrix}^\intercal
\end{equation}
where the state $x$ consists of lateral position $Y$, velocity $v$, and heading angle $\psi$, and the input consists of acceleration $a$ and yaw rate $r$,
The state constraint considered is $\mathcal{C}=\{|Y|\le Y_{\max}\wedge |\psi|\le \psi_{\max}\}$ with $Y_{\max}=1.8m,\psi_{\max}=\pi/3$, and $\mathcal{U}=\{a,r||a|\le a_{\max}\wedge |r|\le r_{\max}\}$.

The backup policy is chosen as
\begin{equation}\label{eq:backup_pol}
  \mathbf{u}(x)=\begin{bmatrix}
    a \\
    r
  \end{bmatrix}=\begin{bmatrix}
                  \Sat_{a_{\max}}(k_v(v_{des}-v)) \\
                  \Sat_{r_{\max}}(k_y[Y;\psi]),
                \end{bmatrix}
\end{equation}
where $\Sat$ is the saturation function, $k_v>0$ is a constant and $k_y$ is calculated via LQR. Under the same controller $\pi$, a CBF is synthesized with the following SOS program:
\begin{equation}\label{eq:SOS}
  \begin{aligned}
  \mathop{\max} ~& R\\
  \mathrm{s.t.}~& h(x)-s_1(x)(Q(x)-R^2)\in\Sigma[x]\\
  & -h(x)-s_2(x)(W^2-Y^2)\in\Sigma[x]\\
  &\nabla h \cdot f(x,\mathbf{u}(x))+\alpha(h)-s_3(x)(Q(x)-R_0^2)\in\Sigma[x]\\
  &s_1(x),s_2(x),s_3(x)\in\Sigma[x],
  \end{aligned}
\end{equation}
with $Q(x)\doteq \hat{x}^\intercal Q\hat{x}$, and $\hat{x} =[Y,v-v_{des},\psi]^\intercal$, where $Q$ is a PSD matrix. $\Sigma[x]$ is the set of sum-of-squares polynomial of $x$, $s_1,s_2,s_3$ are SOS multipliers for the Positivstellensatz procedure that enforces positive definiteness of a polynomial on a semialgebraic set. Due to the cross product term in line 2, \eqref{eq:SOS} cannot be directly solved via SOS, a line search is used to find the largest $R$ that renders the SOS program feasible. $R^0$ is a radius that contains the region of state space we are interested in, \eqref{eq:SOS} essentially tries to fit the largest ball of $Q(x)$ inside $\{h(x)\ge0\}$ where $h$ is the CBF that satisfies \eqref{eq:CBF}. For simplicity, $\sin \psi\approx \psi-\psi^3/3$, and the saturation is lifted, which should lead to a larger control invariant set.
\begin{table*}[t]
\centering
\begin{tabular}{c|c|c|c|c|c|c|c}
            & Dimension & HJ Offline & HJ online & SOS offline & SOS online & Backup online integration & Backup CBF QP \\ \cline{1-8}
Dubin's car & 3         & 37.3s      & 0.5ms     & 3.6s        & 0.4ms      & 1.4ms                     & 3.6ms         \\
Quadrotor   & 16        & NA         & NA        & NA          & NA         & 4.92ms                     & 28.33ms
\end{tabular}
\caption{Computation time of backup CBF, CBF based on HJ, and CBF based on SOS}
\label{tab:computation}
\end{table*}

Fig. \ref{fig:SOS_comp} shows the comparison of the CBF computed with HJ, backup CBF, and SOS where the upper plot shows the 3D surface and the lower plot shows the slicing of the set at $v=5m/s$. As expected, Hamilton Jacobi PDE generates the largest control invariant set, the one from backup CBF is smaller, and the SOS one is the smallest since $h$ is restricted to be a polynomial (4-th order in this case).

Obviously, one can choose a better backup policy and greatly improve the result. For example, if we change $v_{des}=0$, and switch to a more aggressive LQR design, the control invariant set computed from backup CBF is almost identical to the one from Hamilton Jacobi, as shown in Fig. \ref{fig:aggressive}.

\subsection{Aeroplane avoidance}
One classic example in safety-critical control is the aeroplane avoidance with the following dynamics:
\begin{equation}\label{eq:air}
  \begin{bmatrix}
    \Delta \dot{X} \\
    \Delta \dot{Y} \\
    \Delta \dot{\psi}
  \end{bmatrix}=\begin{bmatrix}
                  -v_a+v_b \cos(\Delta \Psi)+u \Delta Y \\
                  v_b \sin(\Delta \psi)-u\Delta X \\
                  -u
                \end{bmatrix},
\end{equation}
where $\Delta X$, $\Delta Y$, and $\Delta \psi$ are the difference of $X,Y$ coordinates and heading angles plane $a$ and $b$, $u$ is the turning rate of plane $a$, $v_a$ and $v_b$ are their velocities, assumed to be constant.

The backup CBF approach is not good at handling malicious disturbance, therefore we assume that plane $b$ maintains a fixed orientation. The safety constraint is defined as $\Delta X^2+\Delta Y^2\ge R^2$, where $R$ is the minimum distance to maintain between the two planes. Fig. \ref{fig:air} shows the danger set ($\{x|h(x)\le0\}$) computed with the HJI PDE and the one computed with a simple backup policy:
\begin{equation}\label{eq:backup_air}
  u = -u_{\max} \mathbf{sign}(\Delta Y).
\end{equation}
The two danger sets are almost identical.
\begin{figure}
  \centering
  \includegraphics[width=1\columnwidth]{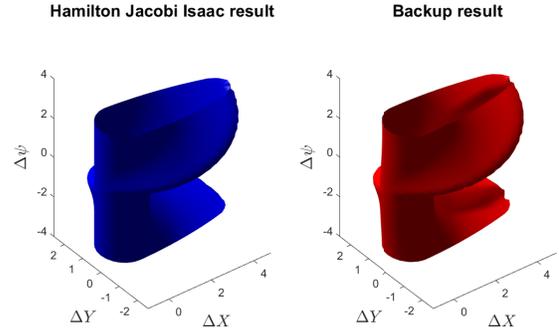}
  \caption{Danger set of the aeroplane collision avoidance}\label{fig:air}
\end{figure}

\subsection{Computation}
Technically, it is not fair to compare the computation time of the backup CBF approach with the benchmarks as the former does not require any offline computation while the online computation is more complicated than normal CBF QP, while CBF based on explicit control invariant set requires computing the control invariant set offline.

Table \ref{tab:computation} shows the computation time of the CBF QP with CBF constructed by the HJ PDE, SOS, and the backup CBF approach. The computation time for backup CBF consists of two parts, the integration of the backup policy, and the CBF QP. As shown by the result, the majority of solver time is spent on the CBF QP, which is 1 magnitude larger than the time used to solve the CBF QP with explicit form. This is mainly because \eqref{eq:CBF_QP_backup} enforces the CBF condition on every $\tau\in[t,t+T]$ (replaced with a sequence of $\tau_i$ in practice), which significantly increases the number of constraints. The integration time is relatively small compared to the QP time and can be further reduced. For instance, \cite{folkestad2020data} uses a Koopman operator approach to further simplify the online integration. We also apply the backup CBF on a 16-dimensional quadrotor model, which is way beyond the limit of HJ (maximum dimension 4-5) and SOS (maximum dimension around 8-10), and the backup CBF can still be implemented with a reasonable loop rate.

\section{Conclusion}
We give a tutorial of the backup CBF approach and compare it to the HJ PDE approach and the SOS approach as benchmarks. The result shows that the implicit control invariant set induced by the backup policy is close to the maximum control invariant set under a properly chosen backup policy in many practical problems. The backup CBF is much more scalable than the benchmark methods with explicitly computed control invariant sets and is applicable to general nonlinear dynamics. Furthermore, since the backup CBF has relative degree 1 under mild assumptions, it is a better choice than the high-order CBFs that bear no feasibility guarantee.

\balance
\bibliographystyle{myieeetran}
\bibliography{mybib}
\end{document}